\pdfoutput=1

\documentclass[sigplan,review=false,nonacm=true]{acmart}

\makeatletter                   
\def\mdseries@tt{m}             
\makeatother                    

\usepackage[plain]{fancyref}
\RequirePackage[frozencache]{minted}
\usepackage{color}
\usepackage{hyperref}           
\hypersetup{
    colorlinks=true,
    linkcolor=blue,
    filecolor=red,      
    urlcolor=magenta,
    breaklinks=true,            
}
\usepackage{breakurl}           

\startPage{1}

\setcopyright{none}

\pdftrailerid{}

\usepackage[utf8]{inputenc}
\usepackage[english]{babel}

\usepackage{hyperref}
\usepackage{xspace}
\usepackage{graphicx}
\usepackage{url}
\usepackage{amsmath}
\usepackage{amssymb}
\usepackage{relsize}
\usepackage{tikz}
\usepackage{clrscode3e}
\usepackage{balance}
\usepackage{amsthm}
\usepackage{color}

\newtheorem{observation}{Observation}

\newcommand{\sm}{\proc{SplitMesher}\xspace}

\microtypecontext{spacing=nonfrench}
\microtypesetup{
  final,
  tracking=true,
  kerning=true,
  spacing=true,
  factor=1100,
  stretch=20,
  shrink=20
}
\SetTracking{encoding={*}, shape=sc}{0} 

\newcommand{\Mesh}{\textsc{Mesh}\xspace}
\newcommand{\mesh}{Mesh\xspace}
\newcommand{\ndegree}{\textup{d}} 

\usepackage{booktabs}   
\RequirePackage[font={normalsize,sf,bf}]{caption}
\RequirePackage[position=bottom]{subfig}
\RequirePackage{float}
\RequirePackage{graphicx}
\RequirePackage{booktabs}
\RequirePackage{multirow}

\begin{document}
\sloppy 

\title[Mesh]{\Mesh: Compacting Memory Management \\ for C/C++ Applications} 


\author{Bobby Powers}
\affiliation{
  \department{College of Information and Computer Sciences}              
  \institution{University of Massachusetts Amherst}            
  \streetaddress{140 Governors Drive}
  \city{Amherst}
  \state{MA}
  \postcode{01003}
  \country{USA}                    
}
\email{bpowers@cs.umass.edu}          

\author{David Tench}
\affiliation{
  \department{College of Information and Computer Sciences}              
  \institution{University of Massachusetts Amherst}            
  \streetaddress{140 Governors Drive}
  \city{Amherst}
  \state{MA}
  \postcode{01003}
  \country{USA}                    
}
\email{dtench@cs.umass.edu}          

\author{Emery D. Berger}
\affiliation{
  \department{College of Information and Computer Sciences}              
  \institution{University of Massachusetts Amherst}            
  \streetaddress{140 Governors Drive}
  \city{Amherst}
  \state{MA}
  \postcode{01003}
  \country{USA}                    
}
\email{emery@cs.umass.edu}          

\author{Andrew McGregor}
\affiliation{
  \department{College of Information and Computer Sciences}              
  \institution{University of Massachusetts Amherst}            
  \streetaddress{140 Governors Drive}
  \city{Amherst}
  \state{MA}
  \postcode{01003}
  \country{USA}                    
}
\email{mcgregor@cs.umass.edu}          

\begin{abstract}

Programs written in C/C++ can suffer from serious memory
fragmentation, leading to low utilization of memory, degraded
performance, and application failure due to memory exhaustion. This
paper introduces \Mesh, a plug-in replacement for \texttt{malloc}
that, for the first time, eliminates fragmentation in unmodified
C/C++ applications. \Mesh combines novel randomized algorithms with
widely-supported virtual memory operations to provably reduce
fragmentation, breaking the classical Robson bounds with high
probability. \Mesh generally matches the runtime performance of
state-of-the-art memory allocators while reducing memory consumption;
in particular, it reduces the memory of consumption of Firefox by 16\%
and Redis by 39\%.

\end{abstract}



\maketitle


\begin{figure*}[t!]
  \centering
  \subfloat[\textbf{Before:} these pages are candidates for
      ``meshing'' \newline because their allocated objects do not
      overlap.\vspace{2em}]{
      \includegraphics[width=0.45\textwidth]{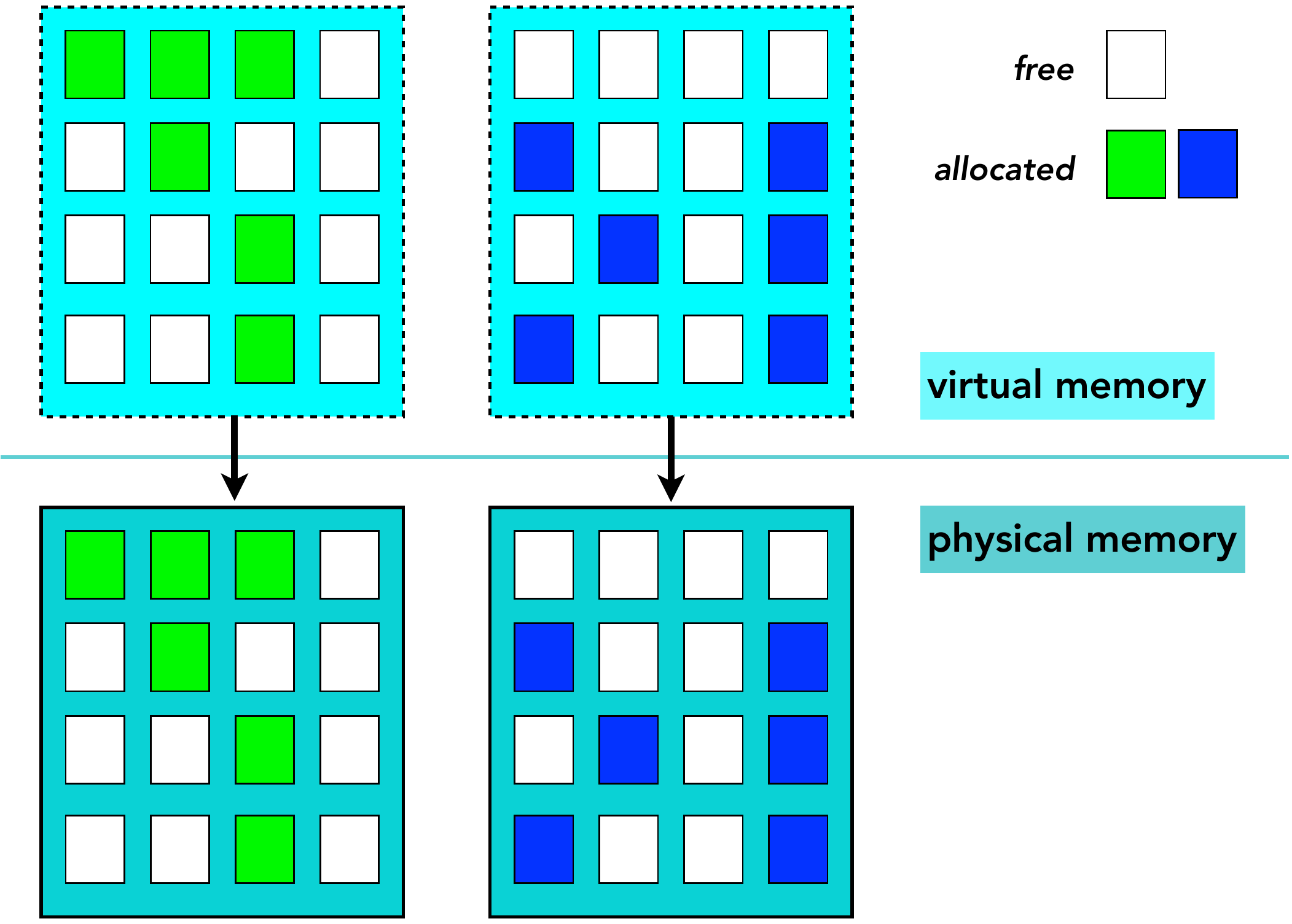}
      \label{pre-meshing}
  }
  ~~~~~
  \centering
  \subfloat[\textbf{After:} both virtual pages now point to the
      first physical page; the second page is now freed.]{
      \includegraphics[width=0.45\textwidth]{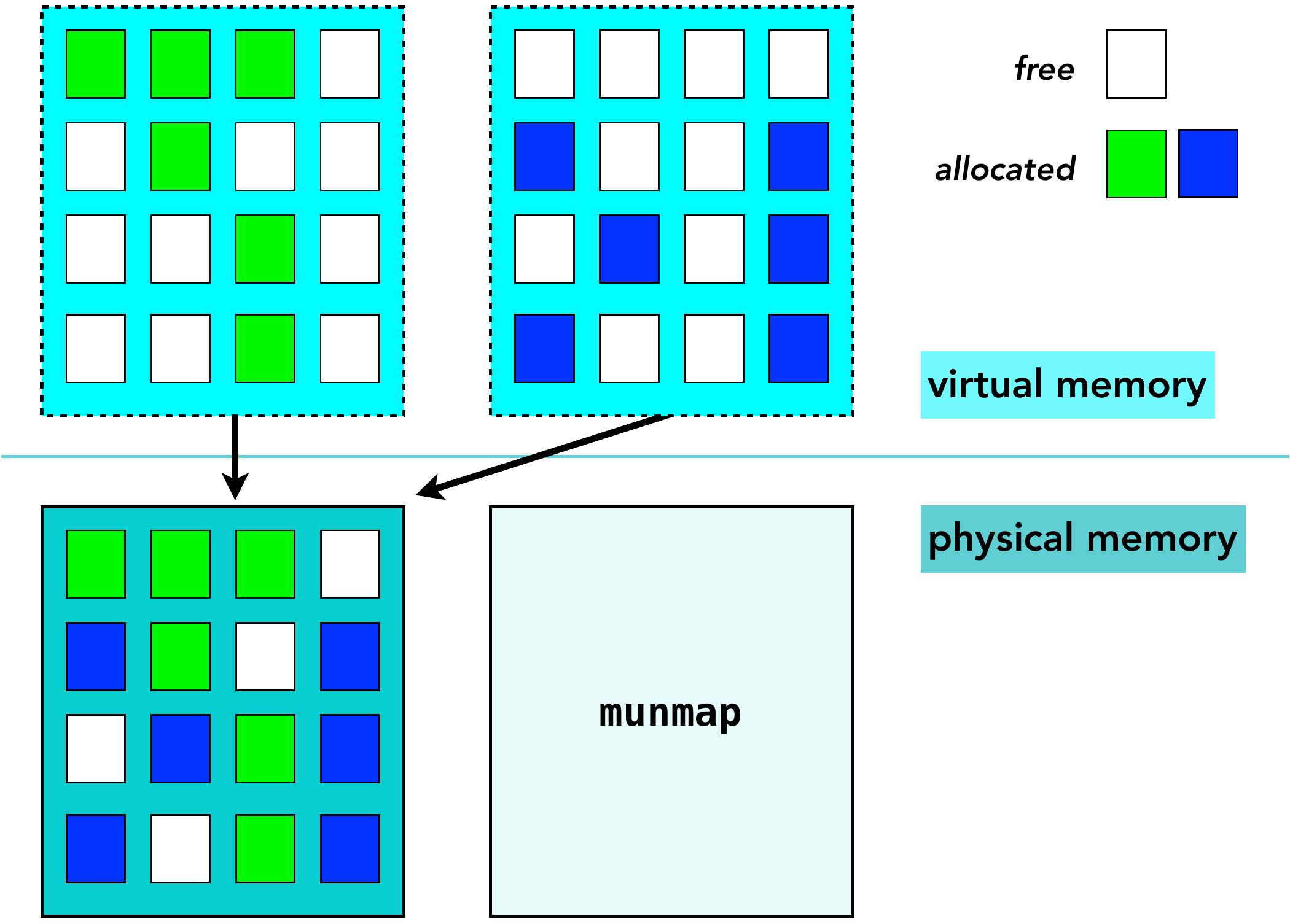}
      \label{post-meshing}
  }

  \caption{\textbf{\Mesh{} in action.} \Mesh{} employs novel
    randomized algorithms that let it efficiently find and then
    ``mesh'' candidate pages within \emph{spans} (contiguous 4K pages)
    whose contents do not overlap.  In this example, it increases
    memory utilization across these pages from 37.5\% to 75\%, and
    returns one physical page to the OS (via \texttt{munmap}),
    reducing the overall memory footprint. \Mesh{}'s randomized
    allocation algorithm ensures meshing's effectiveness with high
    probability.}
  
  \label{fig:meshing}
\end{figure*}

\section{Introduction}
\label{sec:introduction}

Memory consumption is a serious concern across the spectrum of modern
computing platforms, from mobile to desktop to datacenters. For
example, on low-end Android devices, Google reports that more than 99
percent of Chrome crashes are due to running out of memory when
attempting to display a web page~\cite{hara:stateofblink}. On
desktops, the Firefox web browser has been the subject of a five-year
effort to reduce its memory footprint~\cite{awsy}. In datacenters,
developers implement a range of techniques from custom allocators to
other \emph{ad hoc} approaches in an effort to increase memory
utilization~\cite{jemalloc:exposehints,redis:announcement}.

A key challenge is that, unlike in garbage-collected environments,
automatically reducing a C/C++ application's memory footprint
via compaction is not possible. Because the addresses of allocated
objects are directly exposed to programmers, C/C++ applications can
freely modify or hide addresses.  For example, a program may stash
addresses in integers, store flags in the low bits of aligned
addresses, perform arithmetic on addresses and later reference them,
or even store addresses to disk and later reload them.  This hostile
environment makes it impossible to safely relocate objects: if an
object is relocated, all pointers to its original location must be
updated. However, there is no way to safely update \emph{every}
reference when they are ambiguous, much less when they are absent.

Existing memory allocators for C/C++ employ a variety of
best-effort heuristics aimed at reducing average
fragmentation~\cite{johnstone:1998:fragmentation}. However, these
approaches are inherently limited. In a classic result, Robson showed
that all such allocators can suffer from catastrophic
memory fragmentation~\cite{robson:1977:worstcasefrag}. This increase
in memory consumption can be as high as the $\log$ of the ratio
between the largest and smallest object sizes allocated. For example,
for an application that allocates 16-byte and 128KB objects, it is
possible for it to consume $13\times$ more memory than required.

Despite nearly fifty years of conventional wisdom indicating that
compaction is impossible in unmanaged languages, this paper shows that
it is not only possible but also practical. It introduces
\Mesh, a memory allocator that effectively and efficiently performs
compacting memory management to reduce memory usage in unmodified
C/C++ applications.

Crucially and counterintuitively, \Mesh performs compaction without
relocation; that is, without changing the addresses of objects. This
property is vital for compatibility with arbitrary C/C++
applications. To achieve this, \Mesh{} builds on a mechanism which we
call \emph{meshing}, first introduced by Novark et al.'s Hound memory
leak detector~\cite{1542521}. Hound employed meshing in an effort to avoid
catastrophic memory consumption induced by its memory-inefficient
allocation scheme, which can only reclaim memory when every object on
a page is freed. Hound first searches for pages whose live objects do
not overlap. It then copies the contents of one page onto the other,
remaps one of the \emph{virtual} pages to point to the single
\emph{physical} page now holding the contents of both pages, and
finally relinquishes the other physical page to the
OS. Figure~\ref{fig:meshing} illustrates meshing in action.

\Mesh{} overcomes two key technical challenges of meshing that previously made
it both inefficient and potentially entirely ineffective. First,
Hound's search for pages to mesh involves a linear scan of pages on
calls to \texttt{free}. While this search is more efficient than a
naive $O(n^2)$ search of all possible pairs of pages, it remains
prohibitively expensive for use in the context of a general-purpose
allocator. Second, Hound offers no guarantees that \emph{any} pages
would ever be meshable.  Consider an application that happens to
allocate even one object in the same offset in every page. That layout
would preclude meshing altogether, eliminating the possibility of
saving any space.

\Mesh makes meshing both efficient and provably effective (with high
probability) by combining it with two novel randomized
algorithms. First, \Mesh uses a space-efficient randomized
allocation strategy that effectively scatters objects within each
virtual page, making the above scenario provably exceedingly
unlikely. Second, \Mesh incorporates an efficient randomized
algorithm that is guaranteed with high probability to quickly find
candidate pages that are likely to mesh. These two algorithms work in
concert to enable formal guarantees on \Mesh's effectiveness. Our
analysis shows that \Mesh breaks the above-mentioned Robson worst
case bounds for fragmentation with high
probability~\cite{robson:1977:worstcasefrag}.

We implement \Mesh as a library for C/C++ applications running on
Linux of Mac OS X. \Mesh{} interposes on memory management operations,
making it possible to use it without code changes or even
recompilation by setting the appropriate environment variable to load
the \Mesh{} library (e.g., \texttt{export
  LD\_PRELOAD=libmesh.so}). Our empirical evaluation demonstrates that
our implementation of \Mesh{} is both fast and efficient in
practice. It generally matches the performance of state-of-the-art
allocators while guaranteeing the absence of catastrophic
fragmentation with high probability. In addition, it occasionally
yields substantial space savings: replacing the standard allocator
with \Mesh{} automatically reduces memory consumption by 16\%
(Firefox) to 39\% (Redis).

\subsection{Contributions}
\label{sec:contributions}

This paper makes the following contributions:

\begin{itemize}

\item It introduces \textbf{\Mesh}, a novel memory allocator that acts
  as a plug-in replacement for \texttt{malloc}. \Mesh{} combines
  remapping of virtual to physical pages (meshing) with randomized
  allocation and search algorithms to enable safe and effective
  \emph{compaction without relocation} for C/C++
  (\S\ref{sec:meshing}, \S\ref{sec:algorithms},
  \S\ref{sec:allocator}).
  
\item It presents theoretical results that guarantee \Mesh{}'s
    efficiency and effectiveness with high probability (\S\ref{sec:theory}).

\item It evaluates \Mesh{}'s performance empirically, demonstrating \Mesh{}'s ability to reduce
    space consumption while generally imposing low runtime
    overhead (\S\ref{sec:evaluation}).

\end{itemize}

\section{Overview}
\label{sec:meshing}

This section provides a high-level overview of how \Mesh{} works and
gives some intuition as to how its algorithms and implementation
ensure its efficiency and effectiveness, before diving into detailed
description of \Mesh{}'s algorithms (\S\ref{sec:algorithms}),
implementation (\S\ref{sec:allocator}), and its theoretical analysis
(\S\ref{sec:theory}).

\subsection{Remapping Virtual Pages}

\Mesh{} enables compaction without relocating object addresses; it
depends only on hardware-level virtual memory support, which is
standard on most computing platforms like x86 and ARM64. \Mesh{} works
by finding pairs of pages and merging them together \emph{physically}
but not \emph{virtually}: this merging lets it relinquish
physical pages to the OS.

Meshing is only possible when no objects on the pages occupy the same
offsets.  A key observation is that as fragmentation increases (that
is, as there are more free objects), the likelihood of successfully
finding pairs of pages that mesh also increases.

Figure~\ref{fig:meshing} schematically illustrates the meshing
process. \Mesh{} manages memory at the granularity of \textit{spans},
which are runs of contiguous 4K pages (for purposes of illustration,
the figure shows single-page spans). Each span only contains
same-sized objects. The figure shows two spans of memory with low
utilization (each is under $40\%$ occupied) and whose allocations are at non-overlapping offsets.

Meshing consolidates allocations from each span onto one physical span.
Each object in the resulting meshed span resides at the same offset as
it did in its original span; that is, its virtual addresses are
preserved, making meshing invisible to the application. Meshing then
updates the virtual-to-physical mapping (the page tables) for the
process so that both virtual spans point to the same physical
span. The second physical span is returned to the OS.  When average
occupancy is low, meshing can consolidate many pages, offering the
potential for considerable space savings.

\subsection{Random Allocation}

A key threat to meshing is that pages could contain objects at the
same offset, preventing them from being meshed. In the worst case, all
spans would have only one allocated object, each at the same offset,
making them non-meshable. \Mesh{} employs randomized allocation to
make this worst-case behavior exceedingly unlikely. It allocates
objects uniformly at random across all available offsets in a span. As
a result, the probability that all objects will occupy the same offset
is $\left({1}/{b}\right)^{n-1}$, where $b$ is the number of objects in
a span, and $n$ is the number of spans.



In practice, the resulting probability of being unable to mesh many
pages is vanishingly small. For example, when meshing 64 spans with
one 16-byte object allocated on each (so that the number of objects
$b$ in a 4K span is $256$), the likelihood of being unable to mesh any
of these spans is $10^{-152}$. To put this into perspective, there are
estimated to be roughly $10^{82}$ particles in the universe.

We use randomness to guide the design of \Mesh{}'s algorithms
(\S\ref{sec:algorithms}) and implementation (\S\ref{sec:allocator});
this randomization lets us prove robust guarantees of its performance
(\S\ref{sec:theory}), showing that \Mesh{} breaks the Robson bounds with
high probability.



\subsection{Finding Spans to Mesh}

Given a set of spans, our goal is to mesh them in a way that frees as
many physical pages as possible. We can think of this task as that of
partitioning the spans into subsets such that the spans in each subset
mesh. An optimal partition would minimize the number of such subsets.


Unfortunately, as we show, optimal meshing is not feasible
(\S\ref{sec:theory}). Instead, the algorithms in
Section~\ref{sec:algorithms} present practical methods for finding
high-quality meshes under real-world time constraints. We show that
solving a simplified version of the problem (\S\ref{sec:algorithms})
is sufficient to achieve reasonable meshes with high probability
(\S\ref{sec:theory}).

\section{Algorithms}
\label{sec:algorithms}


\Mesh{} comprises three main algorithmic components: allocation
(\S\ref{sec:allocation-algorithm}), deallocation
(\S\ref{sec:deallocation-algorithm}), and finding spans to mesh
(\S\ref{sec:meshing-algorithm}). Unless otherwise noted and without
loss of generality, all algorithms described here are per size class
(within spans, all objects are same size).

\subsection{Allocation}
\label{sec:allocation-algorithm}

Allocation in \Mesh{} consists of two steps: (1) finding a span to
allocate from, and (2) randomly allocating an object from that span.
\Mesh{} always allocates from a thread-local shuffle vector -- a
randomized version of a freelist (described in detail in
\S\ref{sec:shuffle-freelists}). The shuffle vector contains offsets
corresponding to the slots of a single span.  We call that span the
\emph{attached} span for a given thread.

If the shuffle vector is empty, \Mesh relinquishes the current
thread's attached span (if one exists) to the \emph{global heap}
(which holds all unattached spans), and asks it to select a new
span. If there are no partially full spans, the global heap returns a
new, empty span.  Otherwise, it selects a partially full span for
reuse. To maximize utilization, the global heap groups spans into bins
organized by decreasing occupancy (e.g., 75-99\% full in one bin,
50-74\% in the next). The global heap scans for the first non-empty
bin (by decreasing occupancy), and randomly selects a span from that
bin.

Once a span has been selected, the allocator adds the offsets
corresponding to the free slots in that span to the thread-local
shuffle vector (in a random order). \Mesh{} pops the first entry off
the shuffle vector and returns it.

\subsection{Deallocation}
\label{sec:deallocation-algorithm}

Deallocation behaves differently depending on whether the free is
local (the address belongs to the current thread's attached span),
remote (the object belongs to another thread's attached span), or if
it belongs to the global heap.

For local frees, \Mesh{} adds the object's offset onto the span's
shuffle vector in a random position and returns. For remote frees,
\Mesh{} atomically resets the bit in the corresponding index in a
bitmap associated with each span. Finally, for an object belonging to
the global heap, \Mesh{} marks the object as free, updates the span's
occupancy bin; this action may additionally trigger meshing.



\subsection{Meshing}
\label{sec:meshing-algorithm}

When meshing, \Mesh{} randomly chooses pairs of spans and attempts to
mesh each pair. The meshing algorithm, which we call \sm
(Figure~\ref{fig:meshalg}), is designed both for practical
effectiveness and for its theoretical guarantees.  The parameter $t$,
which determines the maximum number of times each span is probed (line
\ref{li:outerloop}), enables space-time trade-offs. The parameter $t$
can be increased to improve mesh quality and therefore reduce space,
or decreased to improve runtime, at the cost of sacrificed meshing
opportunities. We empirically found that $t=64$ balances runtime and
meshing effectiveness, and use this value in our implementation.

\sm proceeds by iterating through $S_l$ and checking whether it can
mesh each span with another span chosen from $S_r$ (line
\ref{li:condition}).  If so, it removes these spans from their
respective lists and meshes them (lines
\ref{li:remove}--\ref{li:mesh}). \sm repeats until it has checked $t *
|S_l|$ pairs of spans; \S\ref{sec:meshing-implementation} describes
the implementation of \sm in detail.

\begin{figure}[!t]
\begin{codebox}
    \Procname{$\sm(S,t)$}
    \li $n \gets$ length$(S)$
    \li $S_l,$ $S_r \gets S[1:n/2],$ $S[n/2 +1 : n]$
    \li \For $(i = 0, i<t, i++)$ \label{li:outerloop}
        \li \Do
            $\mbox{len} = |S_l|$\;
            \li \For $(j = 0, j < \mbox{len}, j++)$ \label{li:innerloop}
                \li \Do
                    \If \proc{Meshable} $(S_l(j)$, $S_r(j+i$ \% $\mbox{len}$)) \Then \label{li:condition}
                        \li $S_l \leftarrow S_l \setminus S_l(j)$ \label{li:remove}
                        \li $S_r \leftarrow S_r \setminus S_r(j+i$ \% $\mbox{len}$)
                        \li \proc{mesh}($ S_l(j)$, $S_r(j+i$ \% $\mbox{len}$)) \label{li:mesh}
                    \End
                \End
        \End
\end{codebox}
\caption{\textbf{Meshing random pairs of spans.} \sm splits the randomly ordered span list $S$ into halves, then probes pairs between halves for meshes.  Each span is probed up to $t$ times.}
\label{fig:meshalg}
\end{figure}

\section{Implementation}
\label{sec:allocator}

We implement \Mesh as a drop-in replacement memory allocator that
implements meshing for single or multi-threaded applications written
in C/C++. Its current implementation work for 64-bit Linux and Mac OS
X binaries. \Mesh can be explicitly linked against by passing
\texttt{-lmesh} to the linker at compile time, or loaded dynamically
by setting the \texttt{LD\_PRELOAD} (Linux) or
\texttt{DYLD\_INSERT\_LIBRARIES} (Mac OS X) environment variables to
point to the \Mesh{} library. When loaded, \Mesh interposes on
standard libc functions to replace all memory allocation functions.

\Mesh combines traditional allocation strategies with meshing to
minimize heap usage.  Like most modern memory
allocators~\cite{Novark:2010:DSH:1866307.1866371,1134000,379232,evans2006scalable,ghemawattcmalloc},
\Mesh is a segregated-fit allocator. \Mesh{} employs fine-grained size
classes to reduce internal fragmentation due to rounding up to the
nearest size class. \Mesh{} uses the same size classes correspond to those
used by jemalloc for objects 1024 bytes and
smaller~\cite{evans2006scalable}, and power-of-two size classes for
objects between 1024 and 16K.  Allocations are fulfilled from the
smallest size class they fit in (e.g., objects of size 33--48 bytes
are served from the 48-byte size class); objects larger than 16K are
individually fulfilled from the global arena.  Small objects are
allocated out of \textit{spans} (\S\ref{sec:meshing}), which are
multiples of the page size and contain between 8 and 256 objects of a
fixed size.  Having at least eight objects per span helps
amortize the cost of reserving memory from the global manager for
the current thread's allocator.

Objects of 4KB and larger are always page-aligned and span at least one
entire page. \Mesh does not consider these objects for meshing;
instead, the pages are directly freed to the OS.

\Mesh's heap organization consists of four main components.
\emph{MiniHeaps} track occupancy and other metadata for spans
(\S\ref{sec:miniheaps}).  \textit{Shuffle vectors} enable efficient,
random allocation out of a MiniHeap (\S\ref{sec:shuffle-freelists}).
\textit{Thread local heaps} satisfy small-object allocation requests
without the need for locks or atomic operations in the common case
(\S\ref{sec:thread-local-heaps}). Finally, the \textit{global heap}
(\S\ref{sec:global-heap}) manages runtime state shared by all threads,
large object allocation, and coordinates meshing operations
(\S\ref{sec:meshing-implementation}).

\subsection{MiniHeaps}
\label{sec:miniheaps}

MiniHeaps manage allocated physical spans of memory and are either
\emph{attached} or \emph{detached}.  An attached MiniHeap is owned by
a specific thread-local heap, while a detached MiniHeap is only
referenced through the global heap.  New small objects are
\textit{only} allocated out of attached MiniHeaps.

Each MiniHeap contains metadata that comprises span length, object
size, allocation bitmap, and the start addresses of any virtual spans
meshed to a unique physical span.  The number of objects that can be
allocated from a MiniHeap bitmap is \textit{objectCount = spanSize /
  objSize}.  The allocation bitmap is initialized to
\textit{objectCount} zero bits.

When a MiniHeap is attached to a
thread-local \emph{shuffle vector} (\S\ref{sec:shuffle-freelists}),
each offset that is unset in the MiniHeap's bitmap is added to the
shuffle vector, with that bit now atomically set to one in the bitmap.
This approach is designed to allow multiple threads to free objects
which keeping most memory allocation operations local in the common
case.

When an object is freed and the free is non-local
(\S\ref{sec:deallocation-algorithm}), the bit is reset.  When a new
MiniHeap is allocated, there is only one virtual span that points to
the physical memory it manages. After meshing, there may be multiple
virtual spans pointing to the MiniHeap's physical memory.

\begin{figure}[!t]
  \centering
  \subfloat[A shuffle vector for a span of size 8, where no objects have
      yet been allocated.]{
      \includegraphics[width=.4\textwidth]{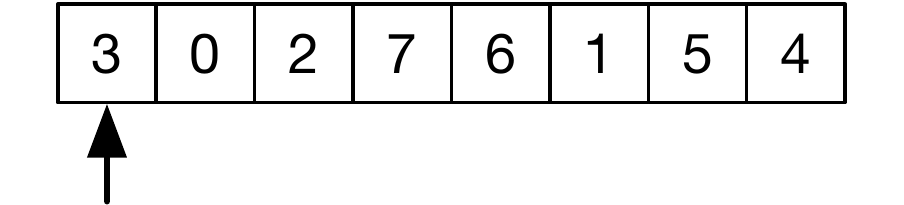}
  }
  \vspace{-0em}
  \subfloat[The shuffle vector after the first object has been allocated.]{
      \includegraphics[width=.4\textwidth]{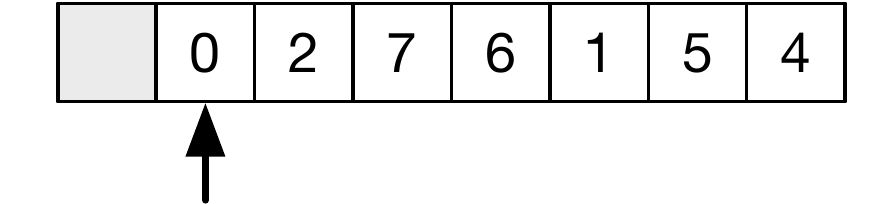}
  }
  \vspace{-0em}
  \subfloat[On \texttt{free}, the object's offset is pushed onto the
      front of the vector, the allocation index is updated, and the
      offset is swapped with a randomly chosen offset.]{
      \includegraphics[width=.4\textwidth]{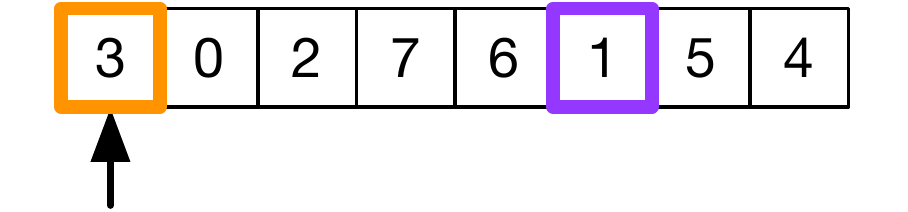}
  }
  \vspace{-0em}
  \subfloat[Finally, after the swap, new allocations proceed in a
      bump-pointer like fashion.]{
    \centering
    \includegraphics[width=.4\textwidth]{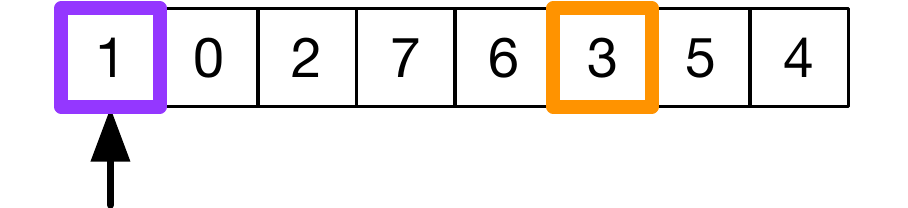}
  }
  \vspace{0.5em}
  \caption{\textbf{Shuffle vectors} compactly enable fast random allocation.
    Indices (one byte each) are maintained in random order; allocation is
    popping, and deallocation is pushing plus a random swap (\S\ref{sec:shuffle-freelists}).}
  \label{fig:shuffle-freelists}
\end{figure}

\subsection{Shuffle Vectors}
\label{sec:shuffle-freelists}

Shuffle vectors are a novel data structure that lets \Mesh
perform randomized allocation out of a MiniHeap efficiently and
with low space overhead.

Previous memory allocators that have employed randomization (for
security or reliability) perform randomized allocation by random
probing into
bitmaps~\cite{1134000,Novark:2010:DSH:1866307.1866371}. In these
allocators, a memory allocation request chooses a random number in the
range $[0,\text{\textit{objectCount}}-1]$. If the associated bit is
zero in the bitmap, the allocator sets it to one and returns the
address of the corresponding offset. If the offset is already one,
meaning that the object is in use, a new random number is chosen and
the process repeated. Random probing allocates objects in $O(1)$
\emph{expected} time but requires overprovisioning memory by a
constant factor (e.g., $2\times$ more memory must be allocated than
needed). This overprovisioning is at odds with our goal of
\emph{reducing} space overhead.

Shuffle vectors solve this problem, combining low space overhead with
worst-case $O(1)$ running time for \texttt{malloc} and
\texttt{free}. Each comprises a fixed-size array consisting of all the
offsets from a span that are \textit{not} already allocated, and an
allocation index representing the head. Each vector is initially
randomized with the Knuth-Fischer-Yates
shuffle~\cite{knuth:1981:semi}, and its allocation index is set to
0. Allocation proceeds by selecting the next available number in the
vector, ``bumping'' the allocation index and returning the
corresponding address. Deallocation works by placing the freed object
at the front of the vector and performing one iteration of the shuffle
algorithm; this operation preserves randomization of the
vector. Figure~\ref{fig:shuffle-freelists} illustrates this process,
while Figure~\ref{fig:malloc} has pseudocode listings for
initialization, allocation, and deallocation.




Shuffle vectors impose far less space overhead than random
probing. First, with a maximum of 256 objects in a span, each offset
in the vector can be represented as an unsigned character (a single
byte). Second, because \Mesh needs only one shuffle vector per
attached MiniHeap, the amount of memory required for vectors is
$256c$, where $c$ is the number of size classes (24 in the current
implementation): roughly 2.8K per thread.  Finally, shuffle vectors
are only ever accessed from a single thread, and so do not require
locks or atomic operations.  While bitmaps must be operated on
atomically (frees may originate at any time from other threads),
shuffle vectors are only accessed from a single thread and do not
require synchronization or cache-line flushes.

\subsection{Thread Local Heaps}
\label{sec:thread-local-heaps}

All malloc and free requests from an application start at the thread's
local heap. Thread local heaps have shuffle vectors for each size
class, a reference to the global heap, and their own thread-local
random number generator.

Allocation requests are handled differently depending on the size of
the allocation.  If an allocation request is larger than 16K, it is
forwarded to the global heap for fulfillment
(\S\ref{sec:global-heap}).  Allocation requests 16K and smaller are
small object allocations and are handled directly by the shuffle
vector for the size class corresponding to the allocation request, as
in Figure~\ref{fig:malloc}a.  If the shuffle vector is empty, it is
refilled by requesting an appropriately sized MiniHeap from the global
heap.  This MiniHeap is a partially-full MiniHeap if one exists, or
represents a freshly-allocated span if no partially full ones are
available for reuse.  Frees, as in Figure~\ref{fig:malloc}d, first
check if the object is from an attached MiniHeap.  If so, it is
handled by the appropriate shuffle vector, otherwise it is passed to
the global heap to handle.

\subsection{Global Heap}
\label{sec:global-heap}

The global heap allocates MiniHeaps for thread-local heaps, handles
all large object allocations, performs non-local frees for both small
and large objects, and coordinates meshing.

\subsubsection{The Meshable Arena}
\label{sec:meshable-arena}

The global heap allocates meshable spans and large objects from a
single, global meshable arena. This arena contains two sets of bins
for same-length spans --- one set is for demand zero-ed spans (freshly
\texttt{mmap}ped), and the other for used spans --- and a mapping of
page offsets from the start of the arena to their owning MiniHeap
pointers.  Used pages are not immediately returned to the OS as they
are likely to be needed again soon, and reclamation is relatively
expensive. Only after 64MB of used pages have accumulated, or whenever
meshing is invoked, \Mesh{} returns pages to OS by calling
\texttt{fallocate} on the heap's file descriptor
(\S\ref{sec:page-table-updates}) with the
\texttt{FALLOC\_FL\_PUNCH\_HOLE} flag.

\begin{figure}[!t]
  \input{all-code}
  \caption{Pseudocode for \Mesh's core allocation and deallocation routines.}
  \label{fig:malloc}
\end{figure}

\subsubsection{MiniHeap allocation}

Allocating a MiniHeap of size $k$ pages begins with requesting $k$
pages from the meshable arena.  The global allocator then allocates
and initializes a new MiniHeap instance from an internal allocator
that \Mesh uses for its own needs. This MiniHeap is kept live so long
as the number of allocated objects remains non-zero, and singleton
MiniHeaps are used to account for large object allocations.  Finally,
the global allocator updates the mapping of offsets to MiniHeaps for
each of the $k$ pages to point at the address of the new MiniHeap.

\subsubsection{Large objects}

All large allocation requests (greater than 16K) are directly
handled by the global heap. Large allocation requests are rounded up
to the nearest multiple of the hardware page size (4K on x86\_64),
and a MiniHeap for 1 object of that size is requested, as detailed
above.  The start of the span tracked by that MiniHeap is returned to
the program as the result of the malloc call.

\subsubsection{Non-local frees}

If \texttt{free} is called on a pointer that is not contained in an
attached MiniHeap for that thread, the free is handled by the global
heap.  Non-local frees occur when the thread that frees the object is
different from the thread that allocated it, or if there have been
sufficient allocations on the current thread that the original
MiniHeap was exhaused and a new MiniHeap for that size class was
attached.

Looking up the owning MiniHeap for a pointer is a constant time
operation. The pointer is checked to ensure it falls within the arena,
the arena start address is subtracted from it, and the result is
divided by the page size.  The resulting offset is then used to index
into a table of MiniHeap pointers. If the result is zero, the
pointer is invalid (memory management errors like double-frees are
thus easily discovered and discarded); otherwise, it points to a live
MiniHeap.

Once the owning MiniHeap has been found, that MiniHeap's bitmap is
updated atomically in a compare-and-set loop.  If a free occurs for an
object where the owning MiniHeap is attached to a different thread,
the free atomically updates that MiniHeap's bitmap, but does not
update the other thread's corresponding shuffle vector.

\subsection{Meshing}
\label{sec:meshing-implementation}

\Mesh's implementation of meshing is guided by theoretical results
(described in detail in Section~\ref{sec:theory}) that enable it to
efficiently find a number of spans that can be meshed.

Meshing is rate limited by a configurable parameter, settable at
program startup and during runtime by the application through the
semi-standard \texttt{mallctl} API.  The default rate meshes at most
once every tenth of a second.  If the last meshing freed less than one
MB of heap space, the timer is not restarted until a subsequent
allocation is freed through the global heap.  This approach ensures that \Mesh
does not waste time searching for meshes when the application and heap
are in a steady state.

We implement the \sm algorithm from Section~\ref{sec:algorithms} in
C++ to find meshes.  Meshing proceeds one size class at a time.  Pairs
of mesh candidates found by \sm are recorded in a list, and after \sm
returns candidate pairs are meshed together \emph{en masse}.

Meshing spans together is a two step process. First, \Mesh{}
consolidates objects onto a single physical span. This consolidation
is straightforward: \Mesh{} copies objects from one span into the free
space of the other span, and updates MiniHeap metadata (like the
allocation bitmap).  Importantly, as \Mesh copies data at the physical
span layer, even though objects are moving in memory, no pointers or
data internal to moved objects or external references need to be
updated. Finally, \Mesh{} updates the process's virtual-to-physical mappings
to point all meshed virtual spans at the consolidated physical
span.

\subsubsection{Page Table Updates}
\label{sec:page-table-updates}

\Mesh updates the process's page tables via calls to \texttt{mmap}.
We exploit the fact that \texttt{mmap} lets the same offset in a file
(corresponding to a physical span) be mapped to multiple
addresses. \Mesh's arena, rather than being an anonymous mapping, as
in traditional \texttt{malloc} implementations, is instead a mapping
backed by a temporary file. This temporary file is obtained via the
\texttt{memfd\_create} system call and only exists in memory or on
swap.

\subsubsection{Concurrent Meshing}

Meshing takes place concurrently with the normal execution of other
program threads with \textit{no} stop-the-world phase required.  This
is similar to how concurrent relocation is implemented in low-latency
garbage collector algorithms like Pauseless and
C4~\cite{click:2005:pauseless, tene:2011:c4}, as described below.
\Mesh maintains two invariants throughout the meshing process: reads
of objects being relocated are always correct and available to
concurrently executing threads, and objects are never written to while
being relocated between physical spans.  The first invariant is maintained
through the atomic semantics of \texttt{mmap}, the second through a
write barrier.

\Mesh's write barrier is implemented with page protections and a
segfault trap handler.  Before relocating objects, \Mesh calls
\texttt{mprotect} to mark the virtual page where objects are being
copied from as read-only.  Concurrent reads succeed as normal.  If a
concurrent thread tries to write to an object being relocated, a
\Mesh-controlled segfault signal handler is invoked by a combination
of the hardware and operating system.  This handler waits on a lock
for the current meshing operation to complete, the last step of which
is remapping the source virtual span as read/write.  Once meshing is
done the handler checks if the address that triggered the segfault was
involved in a meshing operation; if so, the handler exits and the
instruction causing the write is re-executed by the CPU as normal
against the fully relocated object.

\subsubsection{Concurrent Allocation}

All thread-local allocation (on threads other than the one running
\sm) can proceed concurrently and independently with meshing, until
and unless a thread needs a fresh span to allocate from.  Allocation
only is performed from spans owned by a thread, and only spans owned
by the global manager are considered for meshing; spans have a single
owner.  The thread running \sm holds the global heap's lock while
meshing.  This lock also synchronizes transferring ownership of a span
from the global heap to a thread-local heap (or vice-versa).  If
another thread requires a new span to fulfill an allocation request,
the thread waits until the global manager finishes meshing and
releases the lock.

\section{Analysis}
\label{sec:theory}

\newenvironment{claim}[1]{\par\noindent\underline{Claim:}\space#1}{}
\newenvironment{claimproof}[1]{\par\noindent\underline{Proof:}\space#1}{\hfill $\blacksquare$}

\newtheorem{problem}{Problem}

\newcommand{\bigo}{\mathcal{O}}
\newcommand{\page}{\pi}
\newcommand{\str}{s}
\newcommand{\node}{\mathit {v}}
\newcommand{\W}{\mathcal {W}}
\newcommand{\lp}{\left(}
\newcommand{\rparen}{\right)}

This section shows that the \sm procedure described in
\S\ref{sec:meshing-algorithm} comes with strong formal guarantees on
the \textit{quality} of the meshing found along with bounds on its
\textit{runtime}.  In situations where significant meshing
opportunities exist (that is, when compaction is most desirable), \sm
finds with high probability an approximation arbitrarily close to
$1/2$ of the best possible meshing in $O\lp n/q\rparen$ time, where
$n$ is the number of spans and $q$ is the global probability of two
spans meshing.

To formally establish these bounds on quality and runtime, we show
that meshing can be interpreted as a graph problem, analyze its
complexity (\S\ref{subsec:graph}), show that we can do nearly as well
by solving an easier graph problem instead (\S\ref{subsec:matching}),
and prove that \sm approximates this problem with high probability
(\S\ref{subsec:analysis}).


\begin{figure}[!t]
  \centering
  \includegraphics[width=.5\textwidth]{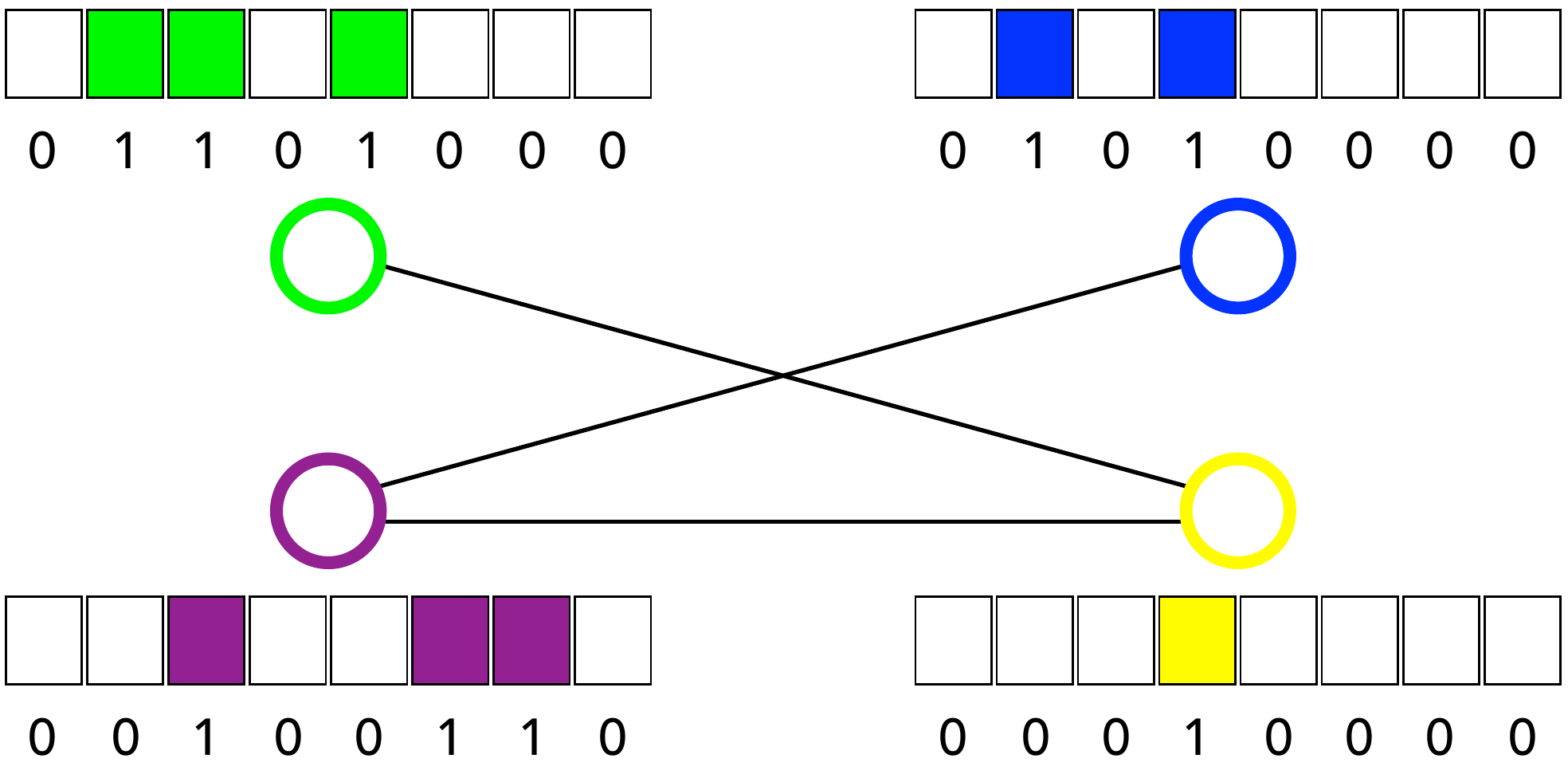}
\caption{\textbf{An example meshing graph.}  Nodes correspond to the spans represented by the strings \texttt{01101000}, \texttt{01010000}, \texttt{00100110}, and \texttt{00010000}.  Edges connect meshable strings (corresponding to non-overlapping spans).}\label{fig:exmesh}
\end{figure}

\subsection{Formal Problem Definitions}
\label{subsec:probdef}
Since \Mesh{} segregates objects based on size, we can limit our
analysis to compaction within a single size class without loss of
generality. For our analysis, we represent spans as binary strings of
length $b$, the maximum number of objects that the span can
store. Each bit represents the allocation state of a single object. We
represent each span $\page$ with string $\str$ such that $\str\lp
i\rparen = 1$ if $\page$ has an object at offset $i$, and 0 otherwise.

\begin{definition}
We say  two  strings $\str_{1}, \str_{2}$ \em{mesh} iff $\sum_i \str_{1}\lp i\rparen \cdot \str_{2} \lp i \rparen = 0$. More generally, a set of binary strings are said to mesh if every pair of strings in this set mesh.
\end{definition}

When we mesh $k$ spans together, the objects scattered across those
$k$ spans are moved to a single span while retaining their offset from
the start of the span. The remaining $k-1$ spans are no longer needed
and are released to the operating system. We say that we ``release''
$k-1$ \emph{strings} when we mesh $k$ strings together.  Since our
goal is to empty as many physical spans as possible, we can
characterize our theoretical problem as follows:

\begin{problem}
Given a multi-set of $n$ binary strings of length $b$, find a meshing
that releases the maximum number of strings.
\end{problem}

Note that the total number of strings released is equal to $n - \rho -
\phi$, where $\rho$ is the number of total meshes performed, and $\phi$
is the number of strings that remain unmeshed.

\paragraph{A Formulation via Graphs:}
\label{subsec:graph}

We observe that an instance of the meshing problem, a string multi-set
$S$, can naturally be expressed via a graph $G(S)$ where there is a
node for every string in $S$ and an edge between two nodes iff the
relevant strings can be meshed. Figure~\ref{fig:exmesh} illustrates
this representation via an example.



If a set of strings are meshable, then there is an edge between every
pair of the corresponding nodes: the set of corresponding nodes is a
\emph{clique}. We can therefore decompose the graph into $k$ disjoint
cliques iff we can free $n-k$ strings in the meshing
problem. Unfortunately, the problem of decomposing a graph into the
minimum number of disjoint cliques (\textsc{MinCliqueCover}) is in
general NP-hard. Worse, it cannot even be approximated up to a factor
$m^{1-\epsilon}$ unless $P=NP$~\cite{zuckerman07}.

While the meshing problem is reducible to {\textsc{MinCliqueCover}},
we have not shown that the meshing problem is NP-Hard.  The meshing
problem is indeed NP-hard for strings of arbitrary length, but in
practice string length is proportional to span size, which is
constant.

\begin{theorem}\label{thm:polytime}
The meshing problem for $S$, a multi-set of strings of constant length, is in $P$. 
\end{theorem}

\begin{proof}
We assume without loss of generality that $S$ does not contain the
all-zero string $\str_0$; if it does, since $\str_0$ can be meshed
with any other string and so can always be released, we can solve the
meshing problem for $S \setminus \str_0$ and then mesh each instance
of $\str_0$ arbitrarily.

Rather than reason about {\textsc{MinCliqueCover}} on a
 meshing graph $G$, we consider the equivalent
problem of coloring the complement graph $\bar{G}$ in which there is
an edge between every pair of two nodes whose strings do not mesh. The nodes of $\bar{G}$ can be partitioned into at most $2^b-1$
subsets $N_1 \ldots N_{2^b-1}$ such that all nodes in each
$N_i$ represent the same string $\str_i$.  The induced subgraph of
$N_i$ in $\bar{G}$ is a clique since all its nodes have a 1 in the same position and so cannot be pairwise meshed.  Further, all nodes in $N_i$ have the same set of neighbors.

Since $N_i$ is a clique, at most one node in $N_i$ may be colored with
any color.  Fix some coloring on $\bar{G}$.  Swapping the colors of
two nodes in $N_i$ does not change the validity of the coloring since
these nodes have the same neighbor set.  We can therefore
unambiguously represent a valid coloring of $\bar{G}$ merely by
indicating in which cliques each color appears.

With $2^b$ cliques and a maximum of $n$ colors, there are at most $\lp
n+1 \rparen ^{c}$ such colorings on the graph where $c=2^{2^b}$. This follows because each color used can be associated with a subset of $\{1, \ldots, 2^b\}$ corresponding to which of the cliques have node with this color; we call this subset a \emph{signature} and note there are $c$ possible signatures. A coloring can be therefore be associated with a multi-set of possible signatures where each signature has multiplicity between 0 and $n$; there are $(n+1)^c$ such multi-sets.
This is polynomial
in $n$ since $b$ is constant and hence $c$ is also constant. So we can simply check each coloring for
validity (a coloring is valid iff no color appears in two cliques
whose string representations mesh).  The algorithm returns a
valid coloring with the lowest number of colors from all valid
colorings discovered.
\end{proof}

Unfortunately, while technically polynomial, the running time of the above algorithm would obviously be prohibitive in practice. Fortunately, as we show, we can exploit the randomness in the strings to design a much faster algorithm.

\subsection{Simplifying the Problem: From \textsc{MinCliqueCover} to \textsc{Matching}}
\label{subsec:matching}
We leverage \Mesh{}'s random allocation to simplify meshing; this random
allocation implies a distribution over the graphs that exhibits useful
structural properties. We first make the following important observation:

\begin{observation}
Conditioned on the occupancies of the strings, edges in the meshing graph  are not three-wise independent.
\end{observation}

To see that edges are not three-wise independent consider three random
strings $s_1, s_2, s_3$ of length 16, each with exactly 6 ones. It is
impossible for these strings to all mesh mutually, so their mesh graph
cannot be a triangle.  Hence, if we know that $s_1$ and $s_2$ mesh,
and that $s_2$ and $s_3$ mesh, we know for certain that $s_1$ and
$s_3$ cannot mesh. More generally, conditioning on $s_1$ and $s_2$
meshing and $s_1$ and $s_3$ meshing decreases the probability that
$s_1$ and $s_3$ mesh.
Below, we quantify this
effect to argue that we can mesh near-optimally by solving the much
easier \textsc{Matching} problem on the meshing graph (i.e.,
restricting our attention to finding cliques of size 2) instead of
\textsc{MinCliqueCover}. Another consequence of the above observation is that we will not be able to appeal to  theoretical results on the standard model of random graphs, \emph{Erd\H{o}s-Renyi graphs}, in which  each possible edge is present with some fixed probability and the edges are fully independent. Instead we will need new algorithms and proofs that only require independence of acyclic collections of edges.

\paragraph{Triangles and Larger Cliques are Uncommon.}
Because of the dependencies across the edges present in a meshing
graph, we can argue that \emph{triangles} (and hence also larger
cliques) are relatively infrequent in the graph and certainly less
frequent than one would expect were all edges independent.  For
example, consider three strings $s_1, s_2, s_3\in \{0,1\}^b$ with
occupancies $r_1, r_2,$ and $r_3$, respectively. The probability they
mesh is
\[
{\binom{b-r_1}{r_2}} \big / {\binom{b}{r_2}} \times {\binom{b-r_{1}-r_2 }{r_3}} \big / {\binom{b}{r_3}} \ . \]

This value is significantly less than would have been the case if the
events corresponding to pairs of strings being meshable were
independent.
For instance, if $b = 32, r_1=r_2=r_3 = 10$, this probability is so
low that even if there were $1000$ strings, the expected number of
triangles would be less than 2. In contrast, had all meshes been
independent, with the same parameters, there would have been $167$ triangles.

The above analysis suggests that we can focus on finding only cliques
of size 2, thereby solving \textsc{Matching} instead of
\textsc{MinCliqueCover}. The evaluation in
Section~\ref{sec:evaluation} vindicates this approach, and we show a
strong accuracy guarantee for \textsc{Matching} below.

\subsection{Theoretical Guarantees}
\label{subsec:analysis}
Since we need to perform meshing at runtime, it is essential that our
algorithm for finding strings to mesh be as efficient as possible. It
would be far too costly in both time and memory overhead to actually
construct the meshing graph and run an existing matching algorithm on
it. Instead, the \sm algorithm (shown in Figure \ref{fig:meshalg})
performs meshing without the need for explicitly constructing the
meshing graph.

For further efficiency, we need to constrain the value of the
parameter $t$, which controls \Mesh{}'s space-time tradeoff. If $t$
were set as large as $n$, then \sm could, in the worst case,
exhaustively search all pairs of spans between the left and right
sets: a total of $n^2/4$ probes.  In practice, we want to choose a
significantly smaller value for $t$ so that \Mesh{} can always
complete the meshing process quickly without the need to search all
possible pairs of strings.

\begin{lemma}
If $t=k/q$ for some user defined parameter $k>1$, \sm finds a matching
of size at least $n(1-e^{-2k})/4$ between the left and right span sets
with probability approaching 1 as $n\geq 2k/q$ grows.
\end{lemma}

\begin{proof}
Let $S_l=\{v_1, v_2, \ldots v_{n/2}\}$ and $S_r=\{u_1,
u_2, \ldots u_{n/2}\}$. Let $t=k/q$ where
$k>1$ is some arbitrary constant. For $u_i\in S_l$ and $i \leq j \leq
j+t$, we say $(u_i,v_j)$ is a \emph{good match} if all the following
properties hold: (1) there is an edge between $u_i$ and $v_j$, (2)
there are no edges between $u_i$ and $v_{j'}$ for $i\leq j'<j$, and
(3) there are no edges between $u_{i'}$ and $v_{j}$ for $i< i'\leq j$.

We observe that \sm finds any good match, although it may
also find additional matches. It therefore suffices to consider only
the number of good matches. The probability $(u_i,v_j)$ is a good
match is $q(1-q)^{2(j-i)}$ by appealing to the fact that the collection of edges under consideration is acyclic. Hence, $\Pr(u_i \mbox{ has a good match})$
is
\begin{align*}
r:= q \sum_{i=0}^{k/q-1} \lp 1-q\rparen ^{2i} = q \frac{1-(1-q)^{2k/q}}{1-(1-q)^2}
 > \frac{1-e^{-2k}}{2} \ .
\end{align*}

To analyze the number of good matches, define $X_i = 1$ iff $u_i$ has
a good match. Then, $\sum_i X_i$ is the number of good matches. By
linearity of expectation, the expected number of good matches is
$rn/2$. We decompose $\sum_i X_i$ into \[Z_0+Z_1+\ldots + Z_{t-1} ~~\mbox{
  where }~~ Z_{j} = \sum_{i\equiv j \bmod t} X_i \ .\] Since each
$Z_j$ is a sum of $n/(2t)$ independent variables, by the Chernoff
bound, $P\lp Z_j < \lp 1-\epsilon \rparen E[Z_j]\rparen \leq \exp\lp -
\epsilon^2 r n/(4t)\rparen$.  By the union bound,
$$P\lp X < \lp 1-\epsilon \rparen rn/2\rparen \leq t \exp\lp - \epsilon^2 r
n/(4t)\rparen$$ and this becomes arbitrarily small as $n$ grows.
\end{proof}

In the worst case, the algorithm checks $nk/2q$ pairs. For our
implementation of \Mesh, we use a static value of $t = 64$; this value
enables the guarantees of Lemma 5.1 in cases where significant meshing
is possible.  As Section~\ref{sec:evaluation} shows, this value for
$t$ results in effective memory compaction with modest performance
overhead.

\subsection{Summary of Analytical Results}
We show the problem of meshing is reducible to a graph problem,
\textsc{MinCliqueCover}.  While solving this problem is infeasible, we
show that probabilistically, we can do nearly as well by finding the
maximum \textsc{Matching}, a much easier graph problem. We analyze our
meshing algorithm as an approximation to the maximum matching on a
random meshing graph, and argue that it succeeds with high
probability.
As a corollary of these results, \Mesh breaks the Robson
bounds with high probability.

\section{Evaluation}
\label{sec:evaluation}


Our evaluation answers the following questions: Does \Mesh reduce
overall memory usage with reasonable performance overhead?
($\S$\ref{sec:evaluation:overallmemory}) Does randomization
provide empirical benefits beyond its analytical guarantees?
($\S$\ref{sec:evaluation:practical})




\subsection{Experimental Setup}
\label{subsec:memory-use}

We perform all experiments on a MacBook Pro with 16 GiB of RAM and an
Intel i7-5600U, running Linux 4.18 and Ubuntu Bionic. We use glibc
2.26 and jemalloc 3.6.0 for SPEC2006, Redis 4.0.2, and Ruby 2.5.1.
Two builds of Firefox 57.0.4 were compiled as release builds, one with
its internal allocator disabled to allow the use of alternate
allocators via \texttt{LD\_PRELOAD}.  SPEC was compiled with clang
version 4.0 at the \texttt{-O2} optimization level, and \Mesh was
compiled with gcc 8 at the \texttt{-O3} optimization level and with
link-time optimization (\texttt{-flto}).


\textbf{Measuring memory usage:} To accurately measure the memory
usage of an application over time, we developed a Linux-based utility,
\texttt{mstat}, that runs a program in a new memory control
group~\cite{redhat:cgroups}. \texttt{mstat} polls the resident-set size
(RSS) and kernel memory usage statistics for all processes in the
control group at a constant frequency.  This enables us to account for
the memory required for larger page tables (due to meshing) in our
evaluation. We have verified that \texttt{mstat} does not perturb
performance results.

\subsection{Memory Savings and Performance Overhead}
\label{sec:evaluation:overallmemory}

We evaluate \Mesh{}'s impact on memory consumption and runtime across
the Firefox web browser, the Redis data structure store, and the
SPECint2006 benchmark suite.

\subsubsection{Firefox}
\label{sec:firefox}

Firefox is an especially challenging application for memory reduction
since it has been the subject of a five year effort to reduce its
memory footprint~\cite{awsy}. To evaluate \Mesh{}'s impact on
Firefox's memory consumption under realistic conditions, we measure
Firefox's RSS while running the Speedometer 2.0 benchmark.
Speedometer was constructed by engineers working on the Google Chrome
and Apple Safari web browsers to simulate the patterns in use on
websites today, stressing a number of browser subsystems like DOM
APIs, layout, CSS resolution and the JavaScript engine.  In Firefox,
most of these subsystems are multi-threaded, even for a single
page~\cite{ff:quantum}.  The benchmark comprises a number of small
``todo'' apps written in a number of different languages and styles,
with a final score computed as the geometric mean of the time taken by
the executed tests.

\begin{figure}[t!]
  \centering
  \vtop{%
    \centering
    \vskip0pt
    \hbox{%
      \includegraphics[width=.45\textwidth]{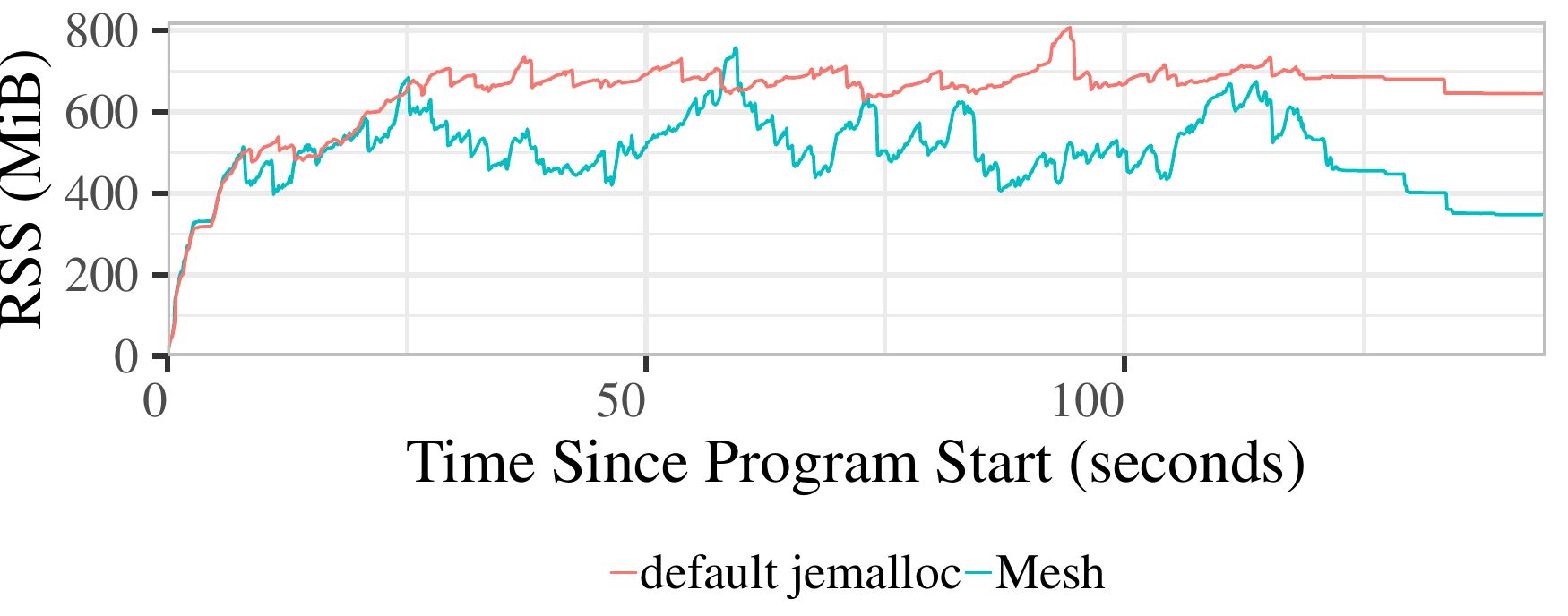}
    }%
  }
  \caption{\textbf{Firefox:} \Mesh decreases mean heap size by 16\%
    over the course of the Speedometer 2.0 benchmark compared with the
    version of jemalloc bundled with Firefox, with less than a 1\%
    change in the reported Speedometer score (\S\ref{sec:firefox}).
    \label{fig:firefox-heap}}
\end{figure}

We test Firefox in single-process mode (disabling content sandboxing,
which spawns multiple processes) under the \texttt{mstat} tool to
record memory usage over time. Our test opens a tab, loads the
Speedometer page from a local server, waits 2 seconds, and then
automatically executes the test.  We record the reported score
at the end of the benchmark run and calculate average memory
usage recorded by \texttt{mstat}.  We tested both a standard release
build of Firefox, along with a release build that did not bundle
Mozilla's fork of jemalloc (hereafter referred to as
\texttt{mozjemalloc}) and instead directly called \texttt{malloc}-related
functions, with \Mesh included via \texttt{LD\_PRELOAD}.  We report
the average resident set size over the course of the benchmark and a
15 second cooldown period afterward, collecting three runs per
allocator.

\Mesh reduces the memory consumption of Firefox by 16\% compared to
Firefox's bundled jemalloc allocator. \Mesh requires 530 MB ($\sigma =
22.4$ MB) to complete the benchmark, while the Mozilla allocator needs
632 MB ($\sigma = 25.3$ MB). This result shows that \Mesh can
effectively reduce memory overhead even in widely used and heavily
optimized applications. \Mesh achieves this savings with less than a 1\%
reduction in performance (measured as the score reported by
Speedometer).

Figure~\ref{fig:firefox-heap} shows memory usage over the course of a
Speedometer benchmark run under \Mesh and the default jemalloc
allocator.  While memory usage under both peaks to similar levels,
Mesh is able to keep heap size consistently lower.

\subsubsection{Redis}
\label{redis-section}

Redis is a widely-used in-memory data structure server.  Redis 4.0
introduced a feature called ``active
defragmentation''~\cite{jemalloc:exposehints,redis:announcement}.
Redis calculates a fragmentation ratio (RSS over sum of active
allocations) once a second.  If this ratio is too high, it triggers a
round of active defragmentation. This involves making a fresh copy of
Redis's internal data structures and freeing the old ones. Active
defragmentation relies on allocator-specific APIs in jemalloc both for
gathering statistics and for its ability to perform allocations that
bypass thread-local caches, increasing the likelihood they will be
contiguous in memory.

\begin{figure}[t!]
  \centering
  \vtop{%
    \centering
    \vskip0pt
    \hbox{%
      \includegraphics[width=.45\textwidth]{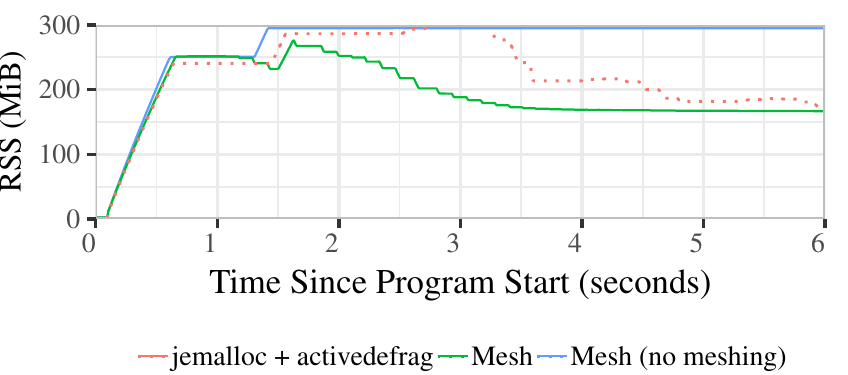}
    }%
  }
  \caption{\textbf{Redis:} \Mesh automatically achieves significant
    memory savings (39\%), obviating the need for its custom,
    application-specific ``defragmentation'' routine
    (\S\ref{redis-section}).
    \label{fig:redis-results}}
\end{figure}




We adapt a benchmark from the official Redis test suite to measure how
\Mesh's automatic compaction compares with Redis's active
defragmentation, as well as against the standard glibc allocator. This
benchmark runs for a total of 7.5 seconds, regardless of allocator. It
configures Redis to act as an LRU cache with a maximum of 100 MB of
objects (keys and values).  The test then allocates 700,000 random
keys and values, where the values have a length of 240 bytes.
Finally, the test inserts 170,000 new keys with values of length 492.
Our only change from the original Redis test is to increase the value
sizes in order to place all allocators on a level playing field with
respect to \emph{internal} fragmentation; the chosen values of 240 and
492 bytes ensure that tested allocators use similar size classes for
their allocations. We test \Mesh with Redis in two configurations:
with meshing always on and with meshing disabled, both without any
input or coordination from the redis-server application.


Figure~\ref{fig:redis-results} shows memory usage over time for Redis
under \Mesh, as well as under jemalloc with Redis's
``activedefrag'' enabled, as measured by \texttt{mstat}
(\S\ref{subsec:memory-use}).  The ``activedefrag'' configuration
enables active defragmentation after all objects have been added to
the cache.

Using \Mesh automatically and portably achieves the same heap size
reduction (39\%) as Redis's active defragmentation.  During most of
the 7.5s of this test Redis is idle; Redis only triggers
active defragmentation during idle periods. With \Mesh, insertion
takes 1.76s, while with Redis's default of jemalloc, insertion takes
1.72s. \Mesh's compaction is additionally
\textit{significantly} faster than Redis's active
defragmentation. During execution with \Mesh{}, a total of 0.23s are
spent meshing (the longest pause is 22 ms), while active
defragmentation accounts for 1.49s ($5.5\times$ slower). This high
latency may explain why Redis disables ``activedefrag'' by
default.


\subsubsection{SPEC Benchmarks}




Most of the SPEC benchmarks are not particularly compelling targets
for \Mesh because they have small overall footprints and do not
exercise the memory allocator. Across the entire SPECint 2006
benchmark suite, \Mesh modestly decreases average memory consumption
(geomean: $-2.4$\%) vs. glibc, while imposing minimal execution time
overhead (geomean: 0.7\%).

However, for allocation-intensive applications with large footprints,
\Mesh is able to substantially reduce peak memory consumption. In
particular, the most allocation-intensive benchmark is
\texttt{400.perlbench}, a Perl benchmark that performs a number of
e-mail related tasks including spam detection (SpamAssassin). With
glibc, its peak RSS is 664MB. \Mesh reduces its peak RSS to 564MB (a
15\% reduction) while increasing its runtime overhead by only
3.9\%.

\subsection{Empirical Value of Randomization}
\label{sec:evaluation:practical}

Randomization is key to \Mesh{}'s analytic guarantees; we evaluate
whether it also can have an observable empirical impact on its ability
to reclaim space. To do this, we test three configurations of \Mesh:
(1) meshing disabled, (2) meshing enabled but randomization disabled,
and (3) \Mesh with both meshing and randomization enabled (the
default).

We tested these configurations with Firefox and Redis, and found no
significant differences when randomization was disabled; we believe
that this is due to the highly irregular (effectively random)
allocation patterns that these applications exhibit. We hypothesized
that a more regular allocation pattern would be more challenging for a
non-randomized baseline. To test this hypothesis, we wrote a synthetic
microbenchmark with a regular allocation pattern in Ruby. Ruby is an
interpreted programming language popular for implementing web
services, including GitHub, AirBnB, and the original version of
Twitter.  Ruby makes heavy use of object-oriented and functional
programming paradigms, making it allocation-intensive.  Ruby is
garbage collected, and while the standard MRI Ruby implementation
(written in C) has a custom GC arena for small objects, large objects
(like strings) are allocated directly with \texttt{malloc}.

Our Ruby microbenchmark repeatedly performs a sequence of string
allocations and deallocations, simulating the effect of accumulating
results from an API and periodically filtering some out. It allocates
a number of strings of a fixed size, then retaining references 25\% of
the strings while dropping references to the rest.  Each iteration the
length of the strings is doubled.  The test requires only a fixed 128
MB to hold the string contents.

\begin{figure}[t!]
  \centering
  \vtop{%
    \centering
    \vskip0pt
    \hbox{%
      \includegraphics[width=.45\textwidth]{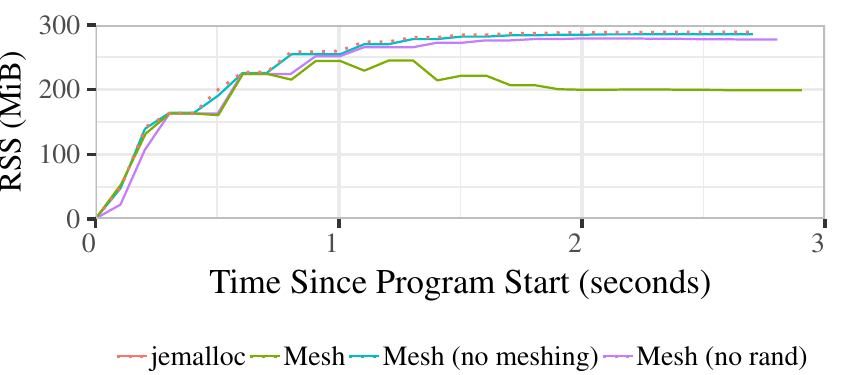}
    }%
  }
  \caption{\textbf{Ruby benchmark:} \Mesh is able to decrease mean heap
    size by 18\% compared to \Mesh with randomization disabled and
    non-compacting allocators ($\S$\ref{sec:evaluation:practical}).
    \label{fig:ruby-frag}}
\end{figure}

Figure~\ref{fig:ruby-frag} presents the results of running this
application with the three variants of \Mesh{} and jemalloc; for this
benchmark, jemalloc and glibc are essentially indistinguishable. With
meshing disabled, \Mesh exhibits similar runtime and heap size to
jemalloc. With meshing enabled but randomization disabled, \Mesh
imposes a 4\% runtime overhead and yields only a modest 3\% reduction in
heap size.

Enabling randomization in \Mesh increases the time overhead to 10.7\%
compared to jemalloc, but the use of randomization lets it
significantly reduce the mean heap size over the execution time of the
microbenchmark (a 19\% reduction). The additional runtime overhead is
due to the additional system calls and memory copies induced by the
meshing process.  This result demonstrates that randomization is not
just useful for providing analytical guarantees but can also be
essential for meshing to be effective in practice.

\subsection{Summary of Empirical Results}

For a number of memory-intensive applications, including aggressively
space-optimized applications like Firefox, \Mesh can substantially
reduce memory consumption (by 16\% to 39\%) while imposing a modest
impact on runtime performance (e.g., around 1\% for Firefox and
SPECint 2006). We find that \Mesh{}'s randomization can enable
substantial space reduction in the face of a regular allocation
pattern.

\section{Related Work}
\label{sec:related-work}

\paragraph{Hound:}
\label{sec:hound}
Hound is a memory leak detector for C/C++ applications that introduced
meshing (a.k.a. ``virtual compaction''), a mechanism that \Mesh{}
leverages~\cite{1542521}. Hound combines an age-segregated heap with
data sampling to precisely identify leaks. Because Hound cannot
reclaim memory until every object on a page is freed, it relies on a
heuristic version of meshing to prevent catastrophic memory
consumption. Hound is unsuitable as a replacement general-purpose
allocator; it lacks both \Mesh's theoretical guarantees and space and
runtime efficiency (Hound's repository is missing files and it does
not build, precluding a direct empirical comparison here). The Hound
paper reports a geometric mean slowdown of $\approx 30\%$ for
SPECint2006 (compared to \Mesh{}'s 0.7\%), slowing one benchmark
(\texttt{xalancbmk}) by almost $10\times$. Hound also generally
\emph{increases} memory consumption, while \Mesh often substantially
decreases it.


\paragraph{Compaction for C/C++:}
Previous work has described a variety of manual and compiler-based
approaches to support compaction for C++. Detlefs shows that if
developers use annotations in the form of smart pointers, C++ code can
also be managed with a relocating garbage
collector~\cite{detlefs:1992:gc}.  Edelson introduced GC support
through a combination of automatically generated smart pointer classes
and compiler transformations that support relocating
GC~\cite{edelson:1992:precompilingcgc}. Google's Chrome uses an
application-specific compacting GC for C++ objects called Oilpan that
depends on the presence of a single event
loop~\cite{google:oilpan}. Developers must use a variety of smart
pointer classes instead of raw pointers to enable GC and
relocation. This effort took years. Unlike these approaches, \Mesh is
fully general, works for unmodified C and C++ binaries, and does not
require programmer or compiler support; its compaction approach is
orthogonal to GC.

CouchDB and Redis implement \emph{ad hoc} best-effort compaction,
which they call ``defragmentation''.  These work by iterating through
program data structures like hash tables, copying each object's
contents into freshly-allocated blocks (in the hope they will be
contiguous), updating pointers, and then freeing the old
objects~\cite{jemalloc:exposehints,redis:announcement}. This
application-specific approach is not only inefficient (because it may
copy objects that are already densely packed) and brittle (because it
relies on internal allocator behavior that may change in new
releases), but it may also be ineffective, since the allocator cannot
ensure that these objects are actually contiguous in memory. Unlike
these approaches, \Mesh performs compaction efficiently and its
effectiveness is guaranteed.

\paragraph{Compacting garbage collection in managed languages:}
Compacting garbage collection has long been a feature of languages
like LISP and
Java~\cite{hansen:1969:compaction,fenichel:1969:compaction}. Contemporary
runtimes like the Hotspot JVM~\cite{microystems2006memory}, the .NET
VM~\cite{microsoft:dotnet-gc}, and the SpiderMonkey JavaScript
VM~\cite{mozilla:spidermonkey-compaction} all implement compaction as
part of their garbage collection algorithms. \Mesh{} brings the
benefits of compaction to C/C++; in principle, it could also be used
to automatically enable compaction for language implementations that
rely on non-compacting collectors.

\paragraph{Bounds on Partial Compaction:}
Cohen and Petrank prove upper and lower bounds on defragmentation via
partial compaction~\cite{Cohen:2017:LPC:3050768.2994597,
  Cohen:2013:LPC:2491956.2491973}. In their setting, corresponding to
managed environments, \emph{every} object \emph{may} be relocated to
any free memory location; they ask what space savings can be achieved
if the memory manager is only allowed to relocate a bounded number of
objects. By contrast, \Mesh{} is designed for unmanaged languages
where objects \emph{cannot} be arbitrarily relocated.


\paragraph{PCM fault mitigation:}
Ipek \emph{et al.} use a technique similar to meshing to address the
degradation of phase-change memory (PCM) over the lifetime of a
device~\cite{ipek:2010:dynamic-replication}.  The authors introduce
dynamically replicated memory (DRM), which uses pairs of PCM pages
with non-overlapping bit failures to act as a single page of
(non-faulty) storage.  When the memory controller reports a page with
new bit failures, the OS attempts to pair it with a complementary
page. A random graph analysis is used to justify this greedy
algorithm.

DRM operates in a qualitatively different domain than \Mesh.  In DRM,
the OS occasionally attempts to pair newly faulty pages
against a list of pages with static bit failures.  This process is
incremental and local.  In \Mesh, the occupancy of spans in the heap
is more dynamic and much less local. \Mesh solves a full,
non-incremental version of the meshing problem each cycle.
Additionally, in DRM, the random graph describes an error model rather
than a design decision; additionally, the paper's analysis is flawed.
The paper erroneously claims that the resulting graph is a simple
random graph; in fact, its edges are not independent (as we show in
\S\ref{subsec:matching}).  This invalidates the claimed performance
guarantees, which depend on properties of simple random graphs. In
contrast, we prove the efficacy of our original \sm algorithm for
\Mesh using a careful random graph analysis.

\section{Conclusion}
\label{sec:conclusion}

This paper introduces \Mesh{}, a memory allocator that efficiently
performs \textit{compaction without relocation} to save memory for
unmanaged languages.  We show analytically that \Mesh{} provably
avoids catastrophic memory fragmentation with high probability, and
empirically show that \Mesh{} can substantially reduce memory
fragmentation for memory-intensive applications written in C/C++ with
low runtime overhead. We have released \Mesh as an open source
project; it can be used with arbitrary C and C++ Linux and Mac OS X
binaries~\cite{mesh-github}. In future work, we plan to explore
integrating \Mesh{} into language runtimes that do not currently
support compaction, such as Go and Rust.

\begin{acks}
  This material is based upon work supported by the
  \grantsponsor{GS100000001}{National Science
    Foundation}{http://dx.doi.org/10.13039/100000001} under Grant
  No.~\grantnum{GS100000001}{1637536}.  Any opinions, findings, and
  conclusions or recommendations expressed in this material are those
  of the author and do not necessarily reflect the views of the
  National Science Foundation.
\end{acks}

\bibliography{emery,mesh}


\end{document}